\documentclass[a4paper,11pt]{article}
\usepackage{lmodern}
\usepackage[utf8]{inputenc}
\usepackage[T1]{fontenc}
\usepackage[USenglish]{babel}
\usepackage{amsmath,amssymb,amsthm}
\usepackage[shortlabels]{enumitem}

\usepackage{fullpage}

\usepackage[numbers,sort&compress]{natbib}

\usepackage[%
  pdftitle={Complexity of solving a system of difference constraints with variables restricted to a finite set},%
  pdfauthor={Santiago Cifuentes and Francisco J.\ Soulignac and Pablo Terlisky},%
  pdfcreator={},%
  pdfsubject={},%
  pdfkeywords={},%
  colorlinks=true,%
  linkcolor=blue,%
  citecolor=blue,%
  urlcolor=blue]{hyperref}  
\usepackage{doi}

\usepackage{algorithm}
\usepackage[indLines=false]{algpseudocodex}

\newtheorem{theorem}{Theorem}

\algrenewtext{While}[1]{\algorithmicwhile\ #1\textbf{:}}
\algrenewtext{For}[1]{\algorithmicfor\ #1\textbf{:}}
\algrenewtext{ForAll}[1]{\algorithmicforall\ #1\textbf{:}}
\algrenewtext{If}[1]{\algorithmicif\ #1\textbf{:}}
\algrenewcommand\algorithmicrequire{\textbf{Input:}}
\algrenewcommand\algorithmicensure{\textbf{Output:}}

\algnewcommand{\inlineFor}[1]{\State \algorithmicfor\ #1\textbf{:}}
\algnewcommand{\algorithmiclet}{\textbf{let}}
\algnewcommand{\Let}{\algorithmiclet\ }

\newcommand{\UnknownSet}{X}
\newcommand{\Unknown}{x}
\newcommand{\Bound}{b}
\newcommand{\ConstraintSet}{S}
\newcommand{\Constraint}{e}
\newcommand{\Solution}{s}
\newcommand{\Domain}{D}
\newcommand{\Element}{d}

\newcommand{\UnknownCount}{n}
\newcommand{\ConstraintCount}{m}
\newcommand{\DomainCount}{k}

\newcommand{\UnsatisfiedSet}{U}

\begin{document}
\title{Complexity of solving a system of difference constraints with variables restricted to a finite set\thanks{Santiago Cifuentes (\texttt{scifuentes@dc.uba.ar}),  Francisco J.\ Soulignac (\texttt{fsoulign@dc.uba.ar}), Pablo Terlisky (\texttt{terlisky@dc.uba.ar})}}

\author{Santiago Cifuentes \and Francisco J.\ Soulignac \and Pablo Terlisky}

\date{\normalsize Universidad de Buenos Aires. Facultad de Ciencias Exactas y Naturales. Departamento de Computación, Buenos Aires, Argentina\\
CONICET-Universidad de Buenos Aires. Instituto de Ciencias de la Computación (ICC), Buenos Aires, Argentina}

\maketitle

\begin{abstract}
Fishburn developed an algorithm to solve a system of $\ConstraintCount$ difference constraints whose $\UnknownCount$ unknowns must take values from a set with $\DomainCount$ real numbers [Solving a system of difference constraints with variables restricted to a finite set, Inform Process Lett 82 (3) (2002) 143--144]. We provide an implementation of Fishburn's algorithm that runs in $O(\UnknownCount+\DomainCount\ConstraintCount)$ time.
\end{abstract}

\section{Introduction}

A \emph{system of difference constraints} is a pair $(\UnknownSet, \ConstraintSet)$ where $\UnknownSet$ is a set of unknowns $\Unknown_1, \ldots, \Unknown_\UnknownCount$ and $\ConstraintSet$ is a family of $\ConstraintCount$ constraints of the form $\Constraint_{ij}\colon \Unknown_i - \Unknown_j \leq \Bound_{ij}$, for $\Bound_{ij} \in \mathbb{R}$.  Throughout this article, we refer to $\Solution \colon \UnknownSet \to \mathbb{R}$ as being a \emph{solution} to $(\UnknownSet, \ConstraintSet)$.  Moreover, we say that $\Solution$ is \emph{restricted} to $\Domain \subseteq \mathbb{R}$ when $\Solution(\Unknown) \in \Domain$ for every $\Unknown \in \UnknownSet$.   If $\Solution(\Unknown_i) - \Solution(\Unknown_j) \leq \Bound_{ij}$ for $\Constraint_{ij} \in \ConstraintSet$, then $\Solution$ \emph{satisfies} $\Constraint_{ij}$, while if $\Solution$ satisfies every constraint in $\ConstraintSet$, then $\Solution$ is \emph{feasible}.  The system $(\UnknownSet, \ConstraintSet)$ itself is \emph{feasible} (resp.\ restricted to $\Domain \subseteq \mathbb{R}$) when it has a feasible solution (resp.\ restricted to $\Domain$).

It is well known that the feasibility of $(\UnknownSet, \ConstraintSet)$ can be decided in $O(\UnknownCount\ConstraintCount)$ time with the Bellman-Ford algorithm~\cite{CormenLeisersonRivestStein2009}.  The output of Bellman-Ford is either a feasible solution or a subset of constraints that admits no solution.  In~\cite{FishburnIPL2002}, Fishburn proposed a simple algorithm to decide if $(\UnknownSet, \ConstraintSet)$ is feasible restricted to a set $\Domain = \{\Element_1 <  \ldots < \Element_\DomainCount\}$, that outputs a feasible solution $\Solution$ restricted to $\Domain$ in the affirmative case.  Fishburn's algorithm can be rephrased as in Algorithm~\ref{alg:fishburn}.
\begin{algorithm}
\caption{Fishburn's Algorithm}\label{alg:fishburn}
\begin{algorithmic}[1]
    \Require A system of difference constraints $(\UnknownSet, \ConstraintSet)$ and a set of real numbers  $\Domain = \{\Element_1 < \ldots < \Element_\DomainCount\}$
    \Ensure a feasible solution $\Solution$ to $(\UnknownSet, \ConstraintSet)$ restricted to $D$ or $\bot$ if such a solution does not exist
    \State \Let $s$ be a function in $\UnknownSet \to \mathbb{R}$
    \inlineFor{$\Unknown \in \UnknownSet$} $\Solution(\Unknown) := \Element_\DomainCount$\label{fishburn:init}
    \While{($\exists\Constraint_{ij} \in \ConstraintSet$)($\Solution$ does not satisfy $\Constraint_{ij})$}\label{fishburn:search}
        \If{$\Element_1 - \Solution(\Unknown_j) > \Bound_{ij}$}
            \State \Return $\bot$\label{fishburn:fail}
        \EndIf
        \State $\Solution(\Unknown_i) := \max\{\Element \in \Domain \mid \Element - \Solution(\Unknown_j)\leq \Bound_{ij}\}$.\label{fishburn:update}
    \EndWhile
    \State \Return $\Solution$\label{fishburn:success}
\end{algorithmic}
\end{algorithm}

Fishburn observed that loop \ref{fishburn:search}--\ref{fishburn:update} is executed $O(\DomainCount\UnknownCount)$ times, because $\Solution(\Unknown_i)$ is decreased at Step~\ref{fishburn:update}.  Consequently, as a single traversal of $\ConstraintSet$ is enough to find an unsatisfied constraint $\Constraint_{ij}$ at Step~\ref{fishburn:search}, Algorithm~\ref{alg:fishburn} runs in $O(\DomainCount\ConstraintCount\UnknownCount)$ time.

Fishburn did not provide a faster implementation of his algorithm in~\cite{FishburnIPL2002}.  Yet, it is easy to see that the update of $\Solution(\Unknown_i)$ at Step~\ref{fishburn:update} only affects the satisfiability of those constraints $\Constraint_{\ell i} \in \ConstraintSet$, $1 \leq \ell \leq \UnknownCount$.  Hence, Step~\ref{fishburn:search} can be restricted to a small subset of $\ConstraintSet$, improving the efficiency of the algorithm.  In this note we take advantage of this fact to show an implementation of Fishburn's algorithm that runs in $O(\UnknownCount+\DomainCount\ConstraintCount)$ time.

Since its appearance in late 2001, Fishburn's algorithm was applied mainly for the optimization of clock skew in digital circuits~\cite{FishburnITC1990}.  In this domain, the canonical application of Fishburn's algorithm is the clock shift decision problem, introduced in early 2002 by Singh and Brown~\cite{SinghBrown2002}: for a fixed clock period $t$ and a finite set $D$ of clock shifts (i.e., delays), determine is there exists an assignment of clock shifts to the registers of a digital circuit that satisfies the double-clocking and zero-clocking constraints~\cite{FishburnITC1990}.  To solve this problem, Singh and Brown apply a ``discrete version of the Bellman-Ford algorithm'' \cite[Algorithm CSDPcore]{SinghBrown2002}, that is nothing else than a restatement of Fishburn's algorithm.  Even though Singh and Brown explicitly state that Fishburn's algorithm runs in $O(\DomainCount\ConstraintCount\UnknownCount)$ time, they also state that they ``implement the algorithm in such a way that searching for an unsatisfied constraint takes at most $O(|R|)$ time'', where $R$ is the set of unknowns~\cite[p.~124]{SinghBrown2002}.  Interestingly, later works that apply Fishburn's algorithm to the clock shift decision problem mention that its time complexity is $O(\DomainCount\ConstraintCount\UnknownCount)$~\cite[e.g.][]{KohiraTakahashiITFECCS2014,LiLuZhouIToVLSIVS2014,Mashiko2017}.  In particular, an $O(\UnknownCount+\ConstraintCount)$ time algorithm that works only for circuits with $\DomainCount = 2$ registers was developed in~\cite{KohiraTakahashiITFECCS2014}, which is based on a reduction of the problem to $2$-SAT.  Our implementation of Fishburn's algorithm generalizes this result, without requiring an implementation of $2$-SAT, as it runs in $O(\UnknownCount+\ConstraintCount)$ time for every constant $\DomainCount$.

\section{An efficient implementation}

Algorithm~\ref{alg:improved} is the improved version of Algorithm~\ref{alg:fishburn} that we propose.  The main difference between both implementations is that the Algorithm~\ref{alg:improved} keeps a set of unknowns $\UnsatisfiedSet$.  Throughout the lifetime of the algorithm, an unknown $\Unknown_i$ belongs to $\UnsatisfiedSet$ if and only if some constraint $\Constraint_{ij} \in \ConstraintSet$ is not satisfied by the current solution $\Solution$.  Hence, $\UnsatisfiedSet$ gives immediate access to those constraints that are not satisfied by $\Solution$.

\begin{algorithm}
\caption{Improved Fishburn's Algorithm}\label{alg:improved}
\begin{algorithmic}[1]
    \Require A system of difference constraints $(\UnknownSet, \ConstraintSet)$ and a set of real numbers  $\Domain = \{\Element_1 < \ldots < \Element_\DomainCount\}$
    \Ensure a feasible solution $\Solution$ to $(\UnknownSet, \ConstraintSet)$ restricted to $D$ or $\bot$ if such a solution does not exist

    \State \Let $s$ be a function in $\UnknownSet \to \mathbb{R}$
    \inlineFor{$\Unknown \in \UnknownSet$} $\Solution(\Unknown) := \Element_\DomainCount$\label{improved:init}
    \For{$i := 1, \ldots, \UnknownCount$}
        \State \Let $\ConstraintSet_i^+ := \{\Constraint_{ij} \in \ConstraintSet \mid 1 \leq j \leq \UnknownCount\}$\label{improved:neighborhood+}
        \State \Let $\ConstraintSet_i^- := \{\Constraint_{ji} \in \ConstraintSet \mid 1 \leq j \leq \UnknownCount\}$.\label{improved:neighborhood-}
    \EndFor
    \State \Let $\UnsatisfiedSet := \{\Unknown_i \in \UnknownSet \mid (\exists \Constraint \in \ConstraintSet_i^+) (\Solution \text{ does satisfy } \Constraint)\}$.\label{improved:queue init}
    \While{$\UnsatisfiedSet \neq \emptyset$}\label{improved:loop init}
        \State Remove some unknown $\Unknown_i$ from $\UnsatisfiedSet$.\label{improved:queue remove}
        \If{$(\exists \Constraint_{ij} \in \ConstraintSet_i^+)(\Element_1 - \Solution(\Unknown_j) > \Bound_{ij})$}\label{improved:fail condition}
            \State \Return $\bot$\label{improved:fail}
        \EndIf
        \State $\Solution(\Unknown_i) := \max\{\Element \in \Domain \mid (\forall \Constraint_{ij} \in \ConstraintSet_i^+)(\Element - \Solution(\Unknown_j) \leq \Bound_{ij})\}$.\label{improved:update}
        \ForAll{$\Constraint_{ji} \in \ConstraintSet_i^{-}$ not satisfied by $\Solution$}
            \State Insert $\Unknown_j$ into $\UnsatisfiedSet$\label{improved:queue update}
        \EndFor
    \EndWhile
    \State \Return $\Solution$\label{improved:success}
\end{algorithmic}
\end{algorithm}

\begin{theorem}
 Let $(\UnknownSet, \ConstraintSet)$ be a system with $\ConstraintCount$ difference constraints and $\UnknownCount$ unknowns, and $\Domain$ be a set with $\DomainCount$ real numbers $\Element_1 < \ldots < \Element_\DomainCount$.  The following statements are true when Algorithm~\ref{alg:improved} is executed with input $(\UnknownSet, \ConstraintSet)$ and $\Domain \subseteq \mathbb{R}$:
 \begin{enumerate}[a), ref=\alph*, itemsep=.05\baselineskip,topsep=.3\baselineskip]
  \item Algorithm~\ref{alg:improved} stops in $O(\UnknownCount+\DomainCount\ConstraintCount)$ time.\label{thm:complexity}
  \item If Algorithm~\ref{alg:improved} returns a solution $\Solution \neq \bot$, then $\Solution$ is a feasible solution to $(\UnknownSet, \ConstraintSet)$ restricted to $\Domain$.\label{thm:feasible solution}
  \item If $(\UnknownSet, \ConstraintSet)$ is feasible restricted to $\Domain$, then Algorithm~\ref{alg:improved} returns a solution $\Solution \neq \bot$.\label{thm:feasible system}
 \end{enumerate}
\end{theorem}

\begin{proof}
 Before dealing with \ref{thm:complexity}--\ref{thm:feasible system}, we prove that the following statements are true immediately before each execution of Step~\ref{improved:loop init}:
 \begin{enumerate}[(i), itemsep=.05\baselineskip, topsep=.3\baselineskip]
  \item $\Solution$ is a solution to $(\UnknownSet, \ConstraintSet)$ restricted to $\Domain$,\label{inv:solution}
  \item if $\Solution'$ is a feasible solution to $(\UnknownSet, \ConstraintSet)$ restricted to $\Domain$, then $\Solution'(\Unknown) \leq \Solution(\Unknown)$ for every $\Unknown \in \UnknownSet$, and\label{inv:solution greatest}
  \item $\Unknown_j \in \UnsatisfiedSet$ ($1 \leq j \leq \UnknownCount$) if and only if $\Solution$ does not satisfy some constraint in $\ConstraintSet_j^+$.\label{inv:unsatisfied}
 \end{enumerate}
 
  Certainly, \ref{inv:solution}--\ref{inv:unsatisfied} hold immediately before the first execution of Step~\ref{improved:loop init}, because of Steps~\ref{improved:init}--\ref{improved:queue init}.  Statement~\ref{inv:solution} remains true because $\Solution$ is updated only by Step~\ref{improved:update}.  Regarding~\ref{inv:solution greatest}, let $\Solution'$ be a feasible solution to $(\UnknownSet,\ConstraintSet)$ restricted to $\Domain$.  By hypothesis, \ref{inv:solution greatest} is true before Step~\ref{improved:loop init}, thus $\Solution'(\Unknown_j) \leq \Solution(\Unknown_j)$.  Then, taking into account that $\ConstraintSet_i^+$ is precisely the subset of constraints in $\ConstraintSet$ that have $\Unknown_i$ with coefficient~$1$ (Step~\ref{improved:neighborhood+}), it follows that $\Solution'(\Unknown_i) \leq \Bound_{ij} + \Solution'(\Unknown_j) \leq \Bound_{ij} + \Solution(\Unknown_j)$ for every $\Constraint_{ij} \in \ConstraintSet^+_i$. Consequently, $\Solution'(\Unknown_i) \leq \Solution(\Unknown_i)$ after Step~\ref{improved:update} and, therefore, \ref{inv:solution greatest} is true immediately before the next execution of Step~\ref{improved:loop init}.  Finally, regarding~\ref{inv:unsatisfied}, observe that $\Solution$ satisfies every constraint in $\ConstraintSet_i^+$ after Step~\ref{improved:update}, thus \ref{inv:unsatisfied} holds for $j = i$ by Step~\ref{improved:queue remove}.  Moreover, $\Constraint_{j\ell} \in \ConstraintSet_{j}^+$ ($1 \leq \ell \leq \UnknownCount$) is not satisfied by $\Solution$ after Step~\ref{improved:update} if and only if either $\ell \neq i$ and $\Constraint_{j\ell}$ was neither satisfied by $\Solution$ before Step~\ref{improved:update} or $\ell = i$ and $\Constraint_{ji}$ is not satisfied by $\Solution$ after its update at Step~\ref{improved:update}.  By \ref{inv:unsatisfied}, the former happens if and only if $\Unknown_j \in \UnsatisfiedSet$ before Step~\ref{improved:queue update}, while the latter happens if and only if $\Constraint_{j\ell} \in \ConstraintSet_i^-$ (Step~\ref{improved:neighborhood-}) is not satisfied by $\Solution$, in which case $\Unknown_{j}$ is inserted into $\UnsatisfiedSet$ at Step~\ref{improved:queue update}.  Altogether, \ref{inv:unsatisfied} is also true before each execution of Step~\ref{improved:loop init}.\nobreak\hfill\nobreak$\triangle$

 In what follows, we use \ref{inv:solution}--\ref{inv:unsatisfied} to prove \ref{thm:complexity}--\ref{thm:feasible system}.

 \ref{thm:complexity}) Suppose $\Unknown_i$ ($1 \leq i \leq \UnknownCount$) is removed from $\UnsatisfiedSet$ at some execution of Step~\ref{improved:queue remove}.  By~\ref{inv:unsatisfied}, some constraint $\Constraint_{ij} \in \ConstraintSet_i^+$ is not satisfied by $\Solution$, thus $\Solution(\Unknown_i) - \Solution(\Unknown_j) > \Bound_{ij}$.  Consequently, either $\Element_1 - \Solution(\Unknown_j) > \Bound_{ij}$ and Algorithm~\ref{alg:improved} halts at Step~\ref{improved:fail} or $\Solution(\Unknown_i)$ is updated to $\Element$ at Step~\ref{improved:update} for some $\Element_1 \leq \Element < \Solution(\Unknown_i)$.  Whichever the case, $\Unknown_i$ is removed at most $O(\DomainCount)$ times from $\UnsatisfiedSet$, thus Algorithm~\ref{alg:improved} runs for a finite amount of time.

 Regarding the time consumed by Algorithm~\ref{alg:improved}, observe that each iteration of the Loop \ref{improved:loop init}--\ref{improved:queue update} requires $O(|\ConstraintSet_i^+|+|\ConstraintSet_i^-|+ p - q)$ time with a standard implementation, where
 \begin{enumerate}
  \item[$\bullet$] $\Element_p$ is the value of $\Solution(\Unknown_i)$ before the execution of \ref{improved:update},
  \item[$\bullet$] $q = 0$ if the loop reaches Step~\ref{improved:fail} and the algorithm halts, and
  \item[$\bullet$] $q > 0$ and $\Element_q$ is the value assigned to $\Solution(\Unknown_i)$ if Step~\ref{improved:update} is reached by the loop.
 \end{enumerate}
  Then, as Steps~\ref{improved:init}--\ref{improved:queue init} require $O(\UnknownCount + \ConstraintCount)$ time and each $\Unknown \in \UnknownSet$ is removed at most $O(\DomainCount)$ times from $\UnsatisfiedSet$, the Handshaking Lemma implies that the total time consumed by Algorithm~\ref{alg:improved} is:
 \begin{displaymath}
  O\left(\UnknownCount + \ConstraintCount + \sum_{i=1}^{\UnknownCount} \DomainCount\left(|\ConstraintSet_i^+| + |\ConstraintSet_i^-|\right)\right) = O(\UnknownCount+\DomainCount\ConstraintCount).
 \end{displaymath}

 \ref{thm:feasible solution}) Algorithm~\ref{alg:improved} returns a solution $\Solution \neq \bot$ only if Step~\ref{improved:success} is executed, a situation that only occurs if $\UnsatisfiedSet = \emptyset$ when Step~\ref{improved:loop init} is executed for the last time. By~\ref{inv:solution}, $\Solution$ is a solution to $(\UnknownSet, \ConstraintSet)$ restricted to $\Domain$, and, by~\ref{inv:unsatisfied}, $\Solution$ is feasible because $\ConstraintSet = \ConstraintSet_1^+ \cup \ldots \ConstraintSet_\UnknownCount^+$.

 \ref{thm:feasible system}) If Algorithm~\ref{alg:improved} returns $\bot$, then Step~\ref{improved:fail} is reached.  By Step~\ref{improved:fail condition}, $\Element_1 - \Solution(\Unknown_j) > \Bound_{ij}$ for some $\Constraint_{ij} \in \ConstraintSet_{i}^+$, where $\Unknown_i$ is the unknown removed from $\UnsatisfiedSet$ at Step~\ref{improved:queue remove}.  If $\Solution'$ is a solution to $(\UnknownSet, \ConstraintSet)$ restricted to $\Domain$, then either $\Solution'(\Unknown_j) > \Solution(\Unknown_j)$ and $\Solution'$ is not feasible by~\ref{inv:solution greatest} or $\Element_1 \leq \Solution'(\Unknown_j) \leq \Solution(\Unknown_j)$ and $\Solution'$ does not satisfy $\Constraint_{ij}$ because $\Element_1 - \Solution'(\Unknown_j) \geq \Element_1 - \Solution(\Unknown_j) > \Bound_{ij}$.  Whichever the case, $\Solution'$ is not feasible.
\end{proof}

\bibliographystyle{abbrvnat}
\bibliography{biblio}

\end{document}